\documentclass[12pt]{iopart}
\usepackage{cstyle_ed}[2013/11/08]
\usepackage[breaklinks=true]{hyperref}
\usepackage{soul} 
\usepackage{color}
\usepackage{enumitem}

\newcommand{\blk}{\color{black}}

\newtheorem{lem}{Lemma}
\newcommand{\ca}[1]{{\cal {#1}}}

\usepackage{framed}

\newcommand{\comm}[1]{} 
\makeatletter
\def\namedlabel#1#2{\begingroup
   \def\@currentlabel{#2}%
   \phantomsection\label{#1}\endgroup
}
\makeatother

\hypersetup{
  colorlinks   = true, 
  urlcolor     = blue, 
  linkcolor    = blue, 
  citecolor   = red 
}

\begin{document}
\title[Global versus local optimality in feedback-controlled qubit purification]{Global versus local optimality in feedback-controlled qubit purification:
new insights from minimizing \ren entropies}
\author{ Colin~Teo$^{1,2}$, Joshua~Combes$^{3,2}$, Howard~M.~Wiseman$^3$}
\address{$^1$ Centre for Quantum Technologies, National University of Singapore, 3 Science Drive 2, Singapore 117543}
\address{$^2$ Centre for Quantum Computation and Communication Technology (Australian Research Council), Centre for Quantum Dynamics, Griffith University, Brisbane, QLD 4111, Australia}
\address{$^3$ Center for Quantum Information and Control, University of New Mexico, Albuquerque, NM 87131-0001, USA}
\eads{\mailto{colintzw@gmail.com}, \mailto{joshua.combes@gmail.com} and \mailto{H.Wiseman@griffith.edu.au}}

\begin{abstract}
 It was first shown by Jacobs, in 2003, that the process of qubit state purification by continuous measurement of one observable can be enhanced, on average, by unitary feedback control. Here, we quantify this by the reduction in any one of the family of
 \ren entropies $S_\alpha$, with  $0< \alpha < \infty$, at some terminal time, revealing the rich structure of stochastic quantum control even for this simple problem. We generalize Jacobs' original argument, which was for (the unique) impurity measure with a linear evolution map under his protocol, by replacing linearity with convexity, thereby making it applicable to \ren entropies $S_\alpha$ for $\alpha$ in a finite interval about $1$. Even with this generalization, Jacobs' argument fails to identify the surprising fact, which we prove by Bellman's principle of dynamic programming, that his protocol is {\em globally} optimal for all \ren entropies whose decrease is {\em locally} maximized by that protocol. Also surprisingly, even though there is a range of \ren entropies whose decrease is always locally maximized by the null-control protocol, that null-control protocol cannot be shown to be globally optimal in any instance. These results highlight the non-intuitive relation between local and global optimality in stochastic quantum control.
\end{abstract}
\maketitle

\section{Introduction }
The precise manipulation of quantum systems is a necessary requirement in the development of quantum devices. These devices typically, require complex dynamics to be performed in a coherent fashion on the systems to accomplish a desired task. One possible technique to exert the required controls is continuous measurement and feedback control \cite{Bel87,WisMil93b,Doh00,VarBru2007,WisMil10}. This technique has been successfully demonstrated in a variety of systems, ranging from the atomic to the mesoscopic see e.g. \cite{BusRotWil2006,VijMacSli2012,ZhoDotPea2012}.

In continuous feedback control, a system is measured weakly and controls, typically unitary, are applied conditioned on the measurement outcome. With a few notable exceptions, the dynamics of systems subject to continuous weak measurement are stochastic and  non-linear, so closed-form solutions do not generally exist. This makes it difficult to find truly optimal control protocols. Notable exceptions are quantum LQG (Linear Quadratic Gaussian) problems, for which many examples have been considered \cite{DohJac1999,BouEdwBel2005,GouBelSmol2005} and a general formalism exists \cite{WisDoh05,WisMil10}.

 Outside of LQG problems, few stochastic quantum optimal feedback control problems have been solved. The first problem of this kind to be solved was Jacobs' rapid purification problem \cite{Jacobs}, as follows. Given a qubit (Hilbert space dimension $D=2$) in the maximally mixed state $\rho= I/2$, and the ability to perform a continuous diffusive-type measurement of a Pauli operator $Z$, together with arbitrary controlled unitaries, what is the control strategy that maximizes the expected value of the purity  $P=\Tr[\rho^2]$ at some final time? Jacobs' problem, and its solution, has inspired much work, some of which we briefly summarize here. Wiseman and Ralph (WR)  \cite{WR} introduced, and solved, the related problem of finding a control to minimize the mean time of first passage (hitting time) to attain a certain purity \cite{WR}. Wiseman and Bouten \cite{WB} rigorously proved using Bellman's principle that the Jacobs and WR protocols were optimal for their respective cost functions, and also (in Jacobs' fixed-time case) for a cost function using a different measure of purity. Belavkin, Negretti, and M\o lmer confirmed the results of Wiseman and Bouten using other methods~\cite{BelNegMol09}.  Shabani and Jacobs \cite{ShaJac08,JacSha08} used viscosity solutions to find, and prove, the optimality of a purification protocol for a qutrit ($D=3$).
Many bounds have been obtained for higher dimensional systems \cite{ComWisJac08,ComWisJacOCo08,ComWis11b}, and for systems restricted to open-loop control \cite{ComWisSco10}.  Recently, Li \emph{et al.} \cite{LiShaSar13} were able to derive some nice results on the optimality qubit of purification in the presence of imperfections such as measurement inefficiencies, which rigorously captures some of the implications of Reference~\cite{ComWis11a}.

Despite this considerable body of work, there has been a lack of proofs of global optimality in the style of Jacobs~\cite{Jacobs}, which are more intuitive compared to rigorous verification theorems as in Wiseman and Bouten's work \cite{WB}. Also, other measures of purity have not been studied in detail and there are reasons to think  that Jacobs' strategy is not optimal for some measures, such as the log-impurity \cite{WR} and log-minimum-infidelity \cite{ComWisJac08}. In this work we study the problem of minimizing a whole family of purity measures --- the \ren entropies --- for a qubit subject to continuous weak measurement in a quantum feedback scenario. As we will show, the study of this family of control objectives reveals various regimes of optimal protocols and highlights some non-intuitive features of stochastic optimal control. Further, our study reveals new limits to the physical process of extracting entropy from a quantum system using measurement.

This paper is organized as follows. In Section \ref{sec:model}, we elaborate on the measurement model and the control strategies. In Section \ref{sec:JacobsProof}, we first restate Jacobs' proof of optimality, and generalize it to include a larger class of cost functions. We next show in Section \ref{sec:HJB}, the criterion for a protocol to be globally optimal at reducing a particular cost function, and show that, for a class of problems including Jacobs' problem as a special case, global optimality for deterministic protocols is true if and only if the protocol is locally optimal everywhere. In Section \ref{sec:WRNotOptimal}, we proceed to show that the condition of local optimality, contrary to intuition, does not allow a proof of global optimality in the case of the Wiseman-Ralph protocol. (This is, essentially, the null-control protocol, and is non-deterministic.)  We conclude in Sec. \ref{sec:conclude} with a summary and discussion of future research.

\section{Measurement model and {control} strategies} \label{sec:model}
The physical model that we consider is identical to those considered previously in \cite{Jacobs,WR,WB}, which is the continuous measurement of a qubit in the $z$ basis, with the ability to control the Hamiltonian evolution of the qubit. Without loss of generality, we assume that the initial state of the qubit satisfies $\Tr[\rho_0 \sig_y] = 0$. Then, to control the evolution of the qubit, we need only implement Hamiltonian controls in the $y$ basis. The equation of motion for the qubit state matrix, conditioned on the result of a measurement, is then
\begin{eqnarray}
  d \rho &= -i \half{\Omega_t} [\sig_y,\rho] dt + (\sig_z \rho \sig_z - \rho) dt + (\sig_z \rho + \rho \sig_z - 2 \Tr[\rho \sig_z] \rho)dW,
\end{eqnarray}
where $\Omega_t$ is the control input, and $dW$ is the stochastic Wiener increment satisfying $dW^2 = dt$. The Wiener increment is related to the measurement result $dR$ via $dW = dR -  2 \Tr[\rho \sig_z]$. The above equation can be equivalently represented in terms of its Bloch components $(x, y, z) = (\Tr[\rho \sig_x],\Tr[\rho \sig_y],\Tr[\rho \sig_z])$,
\begin{eqnarray}
  dx &= (-2x + \Omega_t z) dt - 2 xz dW \\
  dz &= -\Omega_t x dt + 2(1-z^2)dW
\end{eqnarray}
and $y=0$. Thus the length-squared of the Bloch vector is $r^2 = x^2+z^2$, from which the purity $P$ is defined as $\Tr[\rho^2]=\half{1+r^2}$. It turns out that a more convenient parametrization of this problem is given by the impurity, $L=1-P=\half{1-r^2}$, and the angle which the Bloch vector makes with the $z-$axis, $\theta$, such that $x = r\sin\theta$ and $z=r\cos\theta$. This change of variables gives the following stochastic differential equations (SDEs)
\begin{eqnarray}
  dL &= -4L \Big\{ \big[1-(1-2L)u^2 \big]dt + \sqrt{1-2L}\,u\, dW \Big\},  \\
  d\theta &= \big( \Omega_t + f(\theta,L)\big) dt + g(\theta,L) dW,
\end{eqnarray}
where $u=\cos\theta$ and $f$ and $g$ are functions whose form we will not need. As in earlier work \cite{Jacobs,WR,WB}, we make the simplifying assumption that through $\Omega_t$ we are able to control $\theta$, and hence $u$, directly, which is why the form of $f$ and $g$ is unimportant. Physically, this is a realistic assumption for some solid state systems where the measurement time scale is much slower than the control time scale. 

\subsection{ Control goal }
Now that we have a characterization of {the} system dynamics in terms of  the change in impurity ($dL$) and the control ($u_t$) we now characterize our control goal.
 In the remainder of the article, we will be considering the problem of minimizing the \ren entropies of order
 $\alpha$
 \begin{equation} \label{baratheon}
   \vS_\alpha  (\rho) = \inv{1-\alpha}\ln \Tr(\rho^\alpha),
 \end{equation}
given some initial impurity $L(t_0)=L_0$, and a terminal time $T=t-t_0$. The \ren entropies are defined for $\alpha\geq0$ with range $0\leq S_\alpha(\rho) \leq \log \textrm{rank} (\rho)$, and is zero only for pure states. The consideration of \ren entropies can be physically motivated by noticing that the \ren entropy of order $q$, $\vS_q$ is the $q^{-1}$ derivative of the Gibbs Free energy of the system \cite{baez}. The von Neumann entropy is a special case of the \ren entropy when the order approaches 1, which is well-known to be negative the derivative of the Free energy with respect to temperature. When $\alpha = 2$ the \ren entropy, often called the ``Collision entropy'', is minus the logarithm of the purity: $-\ln (P) = -\ln(1-L)$.

We assume no constraints or costs on the Hamiltonian control; effectively, we allow instantaneous control unitaries. This means we need consider only the single SDE
\begin{equation}
 dL = -4L \Big\{ \big[1-(1-2L)u_t^2 \big]dt + \sqrt{1-2L}\,u_t\, dW \Big\}, \label{eqn:Lsde}
\end{equation}
where $u_t$ may be set arbitrarily in the range $[-1,1]$ without cost. 

It is worth noting that similar stochastic control problems for a qubit have been studied before \cite{Jacobs, WR, WB,ShaJac08, JacSha08, BelNegMol09, LiShaSar13}. Two kinds of problems are typically studied. The first, as stated above, aims to find a control law to minimize the value of the cost function at a fixed terminal time averaged over all realizations. We will call this {\em control goal I}. The second is concerned with finding a control law to minimize the mean time taken to attain a fixed value of the cost function, a time which is sometimes called the mean time of first passage or the expected hitting time. We will call this {\em control goal II}. For the most part the impurity $L=1-P$ has been used as the cost function. (The two exceptions are References~\cite{WB,LiShaSar13} where the negative of the length of the Bloch vector $-\sqrt{1-2L}$ was also used.) Summarizing the main results of this literature: (1) for control goal I Jacobs' control strategy \cite{Jacobs} of $u=0$ (keeping the state unbiased with respect to the measurement basis) has been proven to be optimal using local optimality arguments and linearity~\cite{Jacobs} and later dynamic programming~\cite{WB}, (2) for control goal II the Wiseman-Ralph strategy $u=1$ (keeping the state diagonal in the measurement basis) has been shown to be optimal by dynamic programming~\cite{WB}.

\section{Extending Jacobs' argument: local optimality plus convexity} \label{sec:JacobsProof}
In this section, we first restate Jacobs' intuitive proof of optimality given in Reference\,\cite{Jacobs}. We then extend it  (by replacing linearity with convexity) so that it can be applied to minimizing a subfamily of the family of \ren entropies with nonzero measure.   When a control strategy minimizes its cost function at all times we refer to it as {\em locally optimal}. Local optimality was first heuristically defined in this context, in Reference~\cite[Paragraph 3]{DohJacJun01} for any cost function.

In Reference\,\cite{Jacobs}, Jacobs proves that his protocol is optimal by first showing that it is locally optimal at reducing the impurity everywhere. Then, using linearity of the relevant evolution equation, Jacobs shows that  this  implies global optimality of the protocol.
Before we restate his proof, we first define more rigorously what it means for a protocol to be locally optimal. Although the evolution is continuous, we will follow Jacobs in considering the discrete time problem. We assume that the continuous time problem is divided up into $N$ steps with $N\gg 1$, and define the time step $\delta = T/N$.
\begin{mydef}[General cost function] \label{def:GenCost}
  Let $F : [0,\half 1] \to \ca F$ be a general cost function defined on the interval $\mathcal{F}\subset\mathbb{R}$ such that for any $L,L\,' \in [0,\half 1]$, 
  \begin{eqnarray}
    L < L\,'   \Longrightarrow F(L) < F(L\,'),
  \end{eqnarray}
  which is to say that the cost function is a strictly monotonic increasing function. 
\end{mydef}


\begin{mydef}[Local optimality] \label{def:LocalOpt}
  Let $L_u(t)\in [0,\half 1]$ be the {impurity} at time $t$ under the protocol $u(t)$. A protocol $u^*(t)$ is locally optimal at reducing the function $F[L_u(t)]$, if $\forall \,u(t)\in [-1,1]$, 
  the condition
  \begin{equation}
    E(F[L_{u^*}(t+\delta)]|L(t)) \leq E(F[L_u(t+\delta)]|L(t)),
  \end{equation}
  is satisfied, where $E[.]$ denotes an expectation.
\end{mydef}

\noindent {{\bf Remark 1.} The function $F[.]$ can, of course, be the identity function}.  For this case, we have $F[L] = L$, and it was previously shown in Reference \cite{Jacobs} that  Jacobs' protocol, $u(t) \equiv 0$, is locally optimal. That is, it is optimal for reducing the impurity $L$ in any time step. \\

 Next, we define a function which increments the impurity under some protocol $u$ by one time step.
\begin{mydef}[General increment]
 Let $C_u$ be the increment under the protocol $u$, such that it takes as argument a general cost function $F(L)$ and increments the value of $L$ under the protocol u, \ie
\begin{equation}
  C_u(F(L)) = F(L+\Delta_u L),
\end{equation}
 where $\Delta_u L = L(t+\delta)|_{L(t)=L}-L$.
\end{mydef}

\noindent {{\bf Remark 2.} As is evident from \Eref{eqn:Lsde}, for $u=0$, $L$ evolves deterministically, with $\Delta_{u=0}L = -\eps L$ where $\eps=4 \delta$ is a constant. 
That is, 
\begin{equation}
\ti C \equiv C_{u=0}
\end{equation}
is a function in the usual sense, mapping $\ca F \to \ca F$.

\begin{lem}[Monotonicity of $\ti C$] \label{lem:MonoC}
    The increment function for $u=0$ is monotonic increasing, such that
    \begin{equation}
      F(L) < F(L\,') \Longrightarrow \ti C[ F(L) ] < \ti C[F(L\,')].
     \end{equation}   
\end{lem}
\begin{proof}
    Since the cost function $F$ defined in \ref{def:GenCost} is strictly monotonic, this also implies that if
    \begin{equation}
      F(L) < F(L\,'),
    \end{equation}
    we also have
    \begin{equation}
      L < L\,'.
    \end{equation}
    Next, for the protocol $u=0$, the solution to \Eref{eqn:Lsde} is easily shown to be $L=L_0 e^{-4t}$, which is strictly monotonic. Thus,
    \begin{eqnarray}
      L < L\,' \Longrightarrow L + \Delta_{u=0} L < L\,' + \Delta_{u=0} L\,'.
    \end{eqnarray}
    Using again the monotonicity of $F$ and the definition of $\ti C$, we arrive at the desired result.
\end{proof}

Using the fact that $\ti C$ is locally optimal, by definition~\ref{def:LocalOpt}, for $F$
the identity function, \blk Jacobs proves the following:
\begin{thm}[Jacobs~\cite{Jacobs}] \label{thm:Jacobs}
  For the problem of minimizing the impurity $L$ in some finite time interval $T$, the globally optimal protocol is the repeated application of the protocol $u=0$ at each time step.
\end{thm}
\begin{proof}[Proof {\cite{Jacobs}:}]
The proof of this Theorem is heuristic in that it assumes the continuous-in-time stochastic control problem is well approximated by the discrete-in-time version for $\delta$ sufficiently small.

Since the equation for $dL_u$ is, in general, stochastic, we will denote the possible values of $L_u(t+\delta)$,
given $L(t)=L$, by $\{L_u^i\}$,  where value $i$ occurs with probability $p_i$. Although the SDE (\ref{eqn:Lsde}) implies that the evolution is continuous, we again follow Jacobs in representing it by a discrete variable. Then, starting from the definition of local optimality, $\forall u \in [-1,1]$,
\begin{eqnarray}
 \tilde{C}(L) =   E[\tilde{C}(L)] &\leq E[C_u(L)] = \sum_i p_i \,L_u^i,
\end{eqnarray}
where we have used the fact that $\tilde{C}$ is deterministic. 
Consider the left- and right-most expressions $\ti C(L)$ and $\sum_i p_i \,L_u^i$. As remarked above, $\ti C(L)$ is a function that maps $[0,\half 1] \to [0,\half 1]$. Since both $\ti C(L)$ and $\sum_i p_i \,L_u^i \in [0,\half 1]$, by the monotonicity of $\ti C$, $\exists$ some $L\,' > L$ such that $\sum_i p_i \,L_u^i = \ti C(L\,')$. Using monotonicity of $\ti C$ once more, we arrive at 
\begin{eqnarray}
  \tilde{C}^{2}(L) &\leq \tilde{C}(\sum_i p_i \,L_u^i).
\end{eqnarray}
Jacobs then uses the linearity of the function $\ti C$ to write the RHS as
\begin{eqnarray}
\tilde{C}\Big (\sum_i p_i \,L_u^i \Big ) &= \sum_i p_i \tilde{C}(L_u^i). \label{eqn:Jlinear}
\end{eqnarray}
Then, using local optimality, we have for all $u' \in [-1,1]$,
\begin{eqnarray}
 \tilde{C}\Big (\sum_i p_i \,L_u^i \Big ) &\leq \sum_i p_i E[C_{u'}(L_u^i)], \label{eqn:JLopt}\\
&= \sum_{i,j} p_i p_j {\ }L_{u',u}^{i,j}.
\end{eqnarray}
This procedure can clearly be repeated to cover the entire running time of the protocol to prove that $\tilde{C}^{N}(L)$ is a lower bound on the expected impurity  for  any $N$-step control protocol. Thus, it must be the case that $u(t)\equiv 0$ is a globally optimal protocol in reducing the impurity.
\end{proof}
\subsection{Extension to convex functions}
We notice that Jacobs' proof allows a generalization to convex increment functions, which was not required in Reference~\cite{Jacobs} as the function considered was linear.

\begin{thm} \label{thm:LOandCisGO}
  The protocol $u=0$ is the globally optimal control protocol for the minimization of a cost function $F(L)$, if the increment function $\ti C [F]$ satisfies local optimality
  \begin{equation}
    E[\ti C(F)] \leq E[C_{u}(F)],
  \end{equation}
  and is a convex function of $F$.
\end{thm}
\begin{proof}
  The proof of this Theorem is a straightforward generalization to that of Theorem \ref{thm:Jacobs}, and is likewise heuristic (in approximating continuous time by discrete time).

   Using the condition that $\ti C$ satisfies local optimality for the cost function $F$, and the fact that the protocol $u=0$ is deterministic, we have
  \begin{equation}
      \ti C [F] \leq E[C_{u}(F)] = \sum_i \,p_i \,F_u^i,
  \end{equation}
  where once again we let F be a discrete variable, and we denote the possible values of $F_u(t+\delta)$ given $F(t)=F$, by $\{F_u^i\}$,  where value $i$ occurs with probability $p_i$. As before, we wish to apply $\ti C$ to both the right-most and left expressions. Since $\ca F$ is an interval, $\sum_i \,p_i \,F_u^i \in \ca F$. Further, since $F$ is strictly monotonic, it is also one-to-one, and $\exists$ some $\bar L\in [0,\half 1]$ such that $F(\bar L) = \sum_i \,p_i \,F_u^i$. We can then use Lemma \ref{lem:MonoC} and arrive at 
  \begin{equation}
    \ti C^2 [F] \leq \ti C \pare{\sum_i \,p_i\,F_u^i}.
  \end{equation}
  Since $\ti C$ is convex, by Jensen's inequality \cite{RudinBook}, we have
  \begin{eqnarray}
    \ti C \pare{\sum_i \,p_i\,F_u^i} &\leq \sum_i \,p_i\, \ti C(F_u^i),\\
    &\leq \sum_{i} p_i E[C_u (F_u^i)], \\
    &= \sum_{i,j} p_i p_j \,F_{u,u}^{i,j}.
  \end{eqnarray}
  This can then be repeated to cover the running time of the protocol. Thus, $\ti C^{N}[F]$ is the lower bound on the expected $F(L)$ for any $N$-step protocol.
\end{proof}
We next show that if the cost function $F$ and increment function $\ti C$ satisfy certain elementary properties, Theorem \ref{thm:LOandCisGO} provides easily testable conditions on the optimality of the protocol $u=0$.
\begin{lem} \label{lem:LOandCisGO}
The conditions of Theorem~\ref{thm:LOandCisGO} are satisfied if \begin{enumerate}
\item $F(L)$ is invertible
\item  $F(L)$ satisfies
\begin{equation}
  \mathcal{D} [F(L)] \geq 0, \label{eqn:LocalOpt}
\end{equation}
where, with a prime denoting differentiation with respect to $L$,
the differential operator $\mathcal{D}[\bullet]$ is defined as
\begin{equation} \label{DefcalD}
   \mathcal{D} [f(L)] \equiv   4L(1-2L)\Big(f'(L) + 2L f''(L) \Big),
\end{equation}
\item $F(L)$ satisfies
\begin{equation}
 L\frac{F''(L)}{F'(L)}F''(L)-LF'''(L)-F''(L) \geq 0, \label{eqn:Convex}
\end{equation}
\item  $\tilde{C}[F]$ is twice differentiable (${\cal C}^2)$ in $\mathcal{F}$.
\end{enumerate}
\end{lem}

\begin{proof}
  The proof of Lemma \ref{lem:LOandCisGO} can be broken down into first showing that \Eref{eqn:LocalOpt} implies local optimality of $\ti C$, and then showing that \Eref{eqn:Convex} implies convexity of $\ti C$. Local optimality of a protocol $u^*$ according to Definition \ref{def:LocalOpt} implies that the expected increment of the protocol $u^*$ for every time step is minimum. Since the time step $\delta$ can be made arbitrarily small, this intuitively implies that
  \begin{equation}
  E[dF(L)|u^*,L] \leq E[dF(L)|u,L]. \label{eqn:LOint}
\end{equation}
For a function $F(L)$ which is ${\cal C}^2$ in $L$, it can be shown using \Eref{eqn:Lsde} and \ito calculus that
\begin{eqnarray}
  E[dF] &= E\Bigg[-4L \fd{F}{L} dt + u^2 \mathcal{D}[ F] dt \Bigg],
\end{eqnarray}
where ${\cal D}$ is as defined in \Eref{DefcalD}.
Since the inequality in \Eref{eqn:LOint} becomes an equality for the locally optimal protocol, we need to solve the following minimization problem,
\begin{eqnarray}
  \min_u E[dF] &= -4L \fd{F}{L} + \min_u \big \{u^2 \mathcal{D}[ F] \big \} , \\
  &= -4 L \fd{F}{L} + \min \big\{0,\mathcal{D} [F] \big\}, \label{eqn:LocalOptJorWR}
\end{eqnarray}
which evidently reduces to finding the sign of the function $\mathcal{D} [F]$, since $u\in[-1,1]$. Then, $\mathcal{D} [F]\geq0$ implies that Jacobs' protocol,  $u=0$, is the locally optimal protocol. (Conversely,
 $\mathcal{D} [F]\leq0$ would imply  that the Wiseman-Ralph (WR) protocol \cite{WR},  $u=1$, would be the locally optimal protocol.) Thus, \Eref{eqn:LocalOpt} implies local optimality of the protocol $u=0$.

Since we assume that the function $\ti C$ is $\mathcal{C} ^2$ in $\mathcal{F}$, convexity of $\ti C$ also implies a positive second derivative, \ie
\begin{equation}
  \sd{\ti C[F]}{F} \geq 0.
\end{equation}
Now, the increment function $\ti C$ can be written explicitly as
\begin{eqnarray}
  \tilde{C}(F(L)) &= F(L+\Delta_{u=0}L), \\
  &= F(qL), \label{eqn:ConInt}
\end{eqnarray}
where $q=1-\epsilon$ is a constant and in the second line we have used the fact that Jacobs' protocol is linear, \ie~$\ti C = qL$. Letting $f = F(L)$, \Eref{eqn:ConInt} can be written as
\begin{equation}
  \ti C(f) = F(q F^{-1}(f)).
\end{equation}
With this, the second derivative becomes
\begin{equation}
  \sd{\tilde{C}(f)}{f} = \frac{q}{(F'(L))^2} \Bigg\{ q F''(qL) - F'(qL) \frac{F''(L)}{F'(L)} \Bigg\}, \label{eqn:JConvexMid}
\end{equation}
where we have used the abbreviations $L= F^{-1}(f)$. Now recall that $\epsilon = 4 \delta$ is small, so we can expand \Eref{eqn:JConvexMid} to first order in $\eps$, which gives
\begin{equation}
  \sd{\tilde{C}(f)}{f} \approx \frac{\eps}{(F'(L))^2}\Bigg\{ L\frac{F''(L)}{F'(L)}F''(L)-LF'''(L)-F''(L) \Bigg\}.
\end{equation}
Since the terms outside the curly braces are positive, we need only check for the positivity of the expression within the curly braces. The expression within the curly braces is precisely \Eref{eqn:Convex}. Thus, \Eref{eqn:Convex} implies convexity of $\ti C$.

 Since, \Eref{eqn:LocalOpt} and \Eref{eqn:Convex} imply local optimality and  convexity of $\ti C$ repectively, and Theorem \ref{thm:LOandCisGO} states that local optimality and convexity of $\ti C$ implies global optimality, Equations \eref{eqn:LocalOpt} and \eref{eqn:Convex} implies global optimality.
\end{proof}

\subsection{Application to \ren entropies}
For our problem, the \ren entropy $\vS_\alpha$ (\ref{baratheon}) can be written in terms of
$L$ as
\begin{equation}
  \vS_\alpha(L) = \inv{1-\alpha}\ln\Bigg[\pare{\half{1+\sqrt{1-2L}}}^\alpha+\pare{\half{1-\sqrt{1-2L}}}^\alpha \Bigg].
\end{equation}
Since the impurity $L\in[0,\half{1}]$, the above function satisfies the conditions of definition \ref{def:GenCost} and can be shown to be smooth in $L\in[0,\half{1}]$. Theorem \ref{thm:LOandCisGO} then states that Jacobs' protocol is globally optimal for the cost function $S_\alpha(L)$ if the conditions,
\begin{itemize}
  \item[C1. \namedlabel{con:LocalOpt}{C1}]   Local optimality, \ie~$\mathcal{D}[\vS_\alpha(L)]\geq0$, and
  \item[C2. \namedlabel{con:convex}{C2}]   Convexity, \ie~$L \left[\vS_\alpha''(L)\right]^2/{\vS_\alpha'(L)} -L\vS_\alpha'''(L)-\vS_\alpha''(L) \geq 0$,
\end{itemize}
are satisfied.

 We tested both of these conditions numerically over the range $\alpha\in[0,50]$, and the results are summarized in Figure \ref{fig:JProofOpt}. The light red region is where convexity, \ie~condition \ref{con:convex} is satisfied; the light blue region is where local optimality, \ie~condition \ref{con:LocalOpt} is satisfied, and the dark blue region is when both conditions are simultaneously satisfied.
It is worth noting that the Jacobs-style proof only holds for $\alpha\in [0.823,1.103]$, which is the subset of the dark blue region where both local optimality and convexity are satisfied for all $L$.
This restriction is necessary because, as can be seen in both Eqs. \eref{eqn:Jlinear} and \eref{eqn:JLopt}, a non-optimal protocol $u$, may probabilistically increment the impurity $L_i$ towards $L=\half 1$. Thus, Jacobs' proof holds only when both conditions \ref{con:LocalOpt} and \ref{con:convex} hold for $L\in[0,\half 1]$. \begin{figure}[ht!]
\centering
\includegraphics[width = 0.6 \columnwidth]{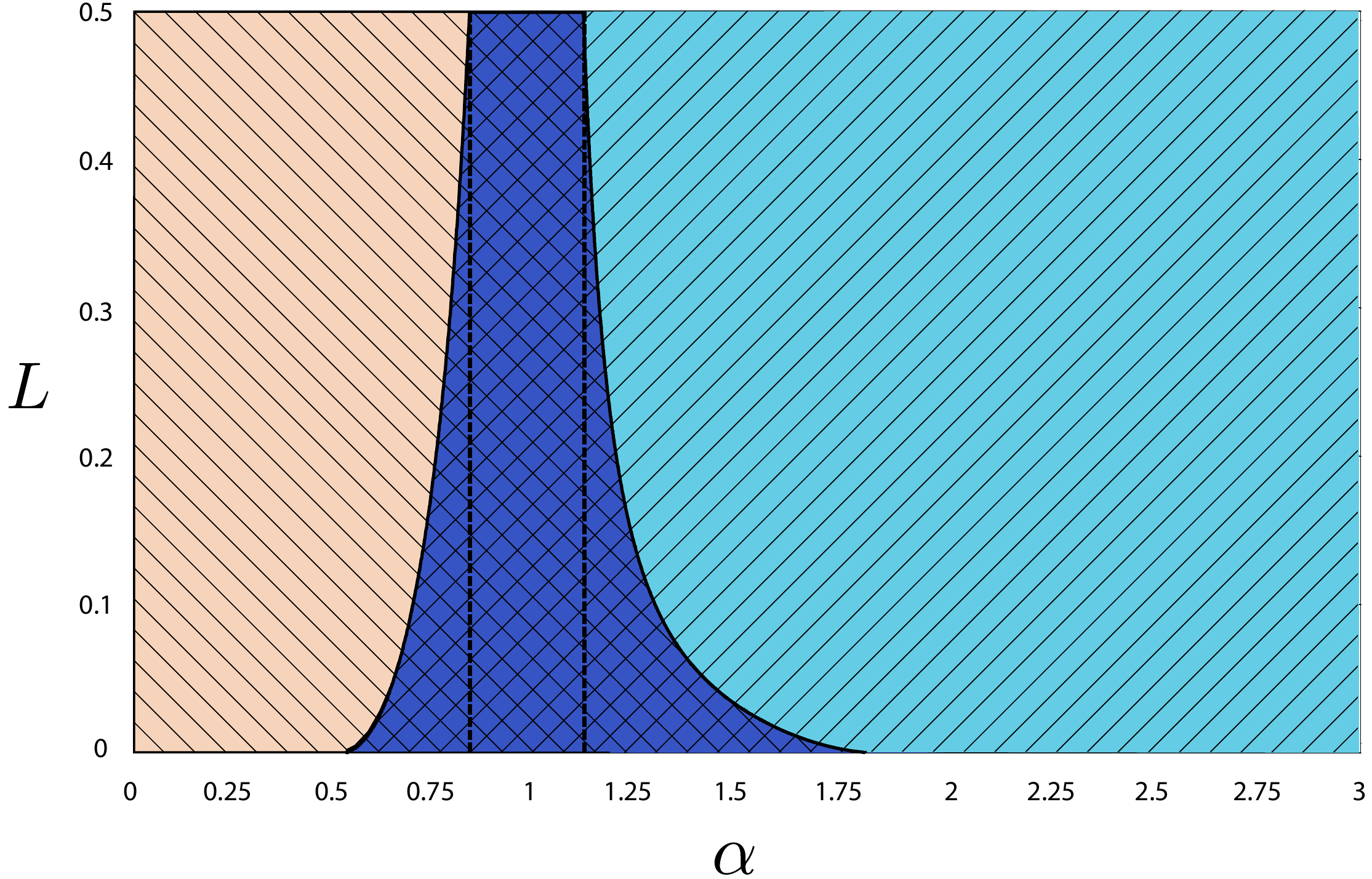}
\caption{The dark blue shaded region is where we have both convexity and local optimality. The light red region is where only convexity is satisfied, and light blue is where only local optimality is satisfied. The x-axis of the graph is the order of the \ren entropy, and the y-axis is the impurity $L$. Then, the region where Jacobs' proof of optimality holds is the region denoted by vertical dashed lines, $\alpha\in[0.823,1.103]$ , since the proof requires that local optimality and convexity be satisfied for all $L\in[0,\half 1]$.} \label{fig:JProofOpt}
\end{figure}

As is evident from Figure \ref{fig:JProofOpt}, the allowed range of $\alpha$ is determined
by examining the conditions for $L=\half 1$. Here,  convexity is satisfied for $\alpha<1.103$ and
local optimality is satisfied for $\alpha\in[0.823,50]$. Furthermore, since the \ren entropy is defined for all $\alpha$, we analytically showed that the min-entropy
\begin{equation}
  \lim_{\alpha\to \infty}\vS_\alpha = -\ln\left (\half{1+\sqrt{1-2L}} \right),
\end{equation}
also satisfies local optimality. Since it is not possible to numerically test all $\alpha>50$, we conjecture that since $\vS_\alpha$ for $\alpha\in [0.823,50] \cup \{\infty\}$ satisfies local optimality, all $\vS_\alpha$ for $\alpha \in [0.823,\infty)$ satisfies local optimality as well. As we will see in the next section, this is relevant for global optimality when we treat the problem using the Bellman's principle of dynamic programming.

\section{Global optimality via Bellman equation} \label{sec:HJB}

In this section, we will derive the Bellman equation for this control problem, and show that for Jacobs' problem, local optimality is equivalent to global optimality.
\\ \\
Let
\begin{equation}
J(u,L,t)=E(F[L_u(T)]|L(t)=L)
 \end{equation}
 be the expected cost function at a final time $T$, given an initial impurity $L$ at time $t\leq T$, with evolution governed by protocol $u \in \ca U$, where the set of allowed control protocols is defined as
 \begin{equation}
     {\cal U} = \{u(t) : [t_0,T] \to [0,1] \}. \label{eqn:SetOfControls}
 \end{equation}
 The globally optimal protocol $u^\star$ (if it exists),  is the protocol which minimizes the above cost function \ie
\begin{equation}
  V(L,t) = J(u^\star,L,t) = \inf_{u\in \ca U} J(u,L,t).
\end{equation}
Where we have denoted $V(L,t)=J(u^\star,L,t)$ as the cost function for the globally optimal protocol. Bellman's principle of dynamic programming then gives,
\begin{equation}
  V(L,t) = \inf_u \{ E(V[L_u(t+h),t+h]) \},
\end{equation}
for $h>0$,  and $t+h\leq T$. The above equation can be intuitively explained as follows. The optimal protocol starting from $t$ must minimize the cost in the interval from $[t,t+h]$, and then proceeding optimally on the time interval $[t+h,T]$, from some impurity $L_u(t+h)$. Dividing the above by $h$ and letting $h\to0$, we arrive at the following Bellman equation,
\begin{eqnarray}
& \fpd{V(L,t)}{t} - 4L \fpd{V(L,t)}{L} +  \inf_u  \Big(u^2 (t)  \ti{\mathcal{D}} [V(L,t)] \Big) \,=\, 0, \\
 \textrm{where }\; &\ti{\mathcal{D}} [V] \,=\, 4L(1-L) \big[ \fpd{V}{L} + 2L \spd{V}{L} \big],
\end{eqnarray}
with the terminal condition $V(L,T) = F(L)$, and we have used the SDE \eref{eqn:Lsde}. This minimization problem is equivalent to
\begin{equation}
  \fpd{V(L,t)}{t} - 4L \fpd{V(L,t)}{L} + \min \{0,  \ti{\mathcal{D}} [ V(L,t)] \}=0, \label{eqn:BellmanEqn}
\end{equation}
since $u\in[-1,1]$.

\subsection{Verification theorem}
The Bellman equation derived in \eref{eqn:BellmanEqn} provides necessary conditions to rigorously show that a control protocol is optimal. This point is best stated by the following theorem:
\begin{thm}[Verification Theorem] \label{thm:verify}
  Let $W(L,t) : [0,\half 1] \times \ca T \to \ca W$, where $\ca T = [t_0,T)$, be a function with the following properties:
    \begin{enumerate}[label=(\alph*)]
      \item W(L,t) is $\ca C^1$ in $\ca T$ and $\ca C^2$ in $[0,\half 1 ]$. \label{con:continuity}
      \item Dynkin's formula:
      \begin{equation}
        E[ W(L(s),s) \big | L,t] - W(L,t) = E\left[ \intto{t}{s} \ca A^u W(L(r),r) dr \,\Big | L(t),t\right],
      \end{equation}
      holds for $t<s<T$, and the operator $\ca A^u$ is defined as
      \begin{equation}
        \ca A^u W(L,t) = \fpd{W(L,t)}{t} - 4L \fpd{W(L,t)}{L} + u^2 \ti D W(L,t),
      \end{equation}
      for some $u \in  \ca U$. \label{con:Dynkin}
      \item W(L,t) is a solution to \Eref{eqn:BellmanEqn}, \ie
  \begin{equation}
    \fpd{W(L,t)}{t} - 4L \fpd{W(L,t)}{L} + \min \{0,  \ti{\mathcal{D}} [ W(L,t)] \}=0, \label{eqn:VeriHypo}
  \end{equation}
  with terminal condition $W(L,T) = F(L)$. \label{con:BellmanAndTerminal}
    \end{enumerate}
Then, the following properties hold:
  \begin{enumerate}
    \item $W(L,t) \leq J(u,L,t)$ for every admissible $u \in \ca U$. \label{prop:InfForAll}
    \item If there exists some $u^* \in \ca U$ such that
    \begin{eqnarray}
      & \fpd{W(L,t)}{t} - 4L \fpd{W(L,t)}{L} + (u^*)^2 \ti{\mathcal{D}} [W(L,t)]  \nn \\
      & =\, \fpd{W(L,t)}{t} - 4L \fpd{W(L,t)}{L} +  \inf_{u\in \ca U}  \Big(u^2 (t)  \ti{\mathcal{D}} [W(L,t)] \Big), \label{eqn:minimum}
    \end{eqnarray}
    then $W(L,t) = V(L,t)$ and $u^*$ is an optimal control protocol. \label{prop:Optimal}
  \end{enumerate}
\end{thm}

For the proof of Theorem \ref{thm:verify} and further details on the verification theorem, the reader is referred to References \cite{FlemingRishelBook} and \cite{FlemingSonerBook}.

 In \Eref{eqn:BellmanEqn}, we are again interested in finding the sign of a function that looks very similar to \Eref{eqn:LocalOpt}. However, it must be noted that in \Eref{eqn:BellmanEqn}, the derivatives are with respect to the initial condition $L(t)$, and not the ``local'' parameters $L$ as in \Eref{eqn:LocalOpt}.
 Nevertheless, as we will show below,  Jacobs' protocol is globally optimal if and only if it is also locally optimal.

\subsection{Equivalence of global optimality and local optimality for a class of problems}
Let us consider a more general class of control problems with the following SDE for the controlled stochastic variable $l$:
\begin{equation}
  dl = - k\, l\Big([1- \beta^2(u,l,t)] \gamma(t) dt+\beta(u,l,t)\sqrt{\gamma(t)} dW\Big), \label{eqn:GenSDE}
\end{equation}
where $u \in \ca M$ is the control {function, chosen from an appropriate set of controls $\ca M$}.  We impose the following constraints on the real functions $\gamma(t),$ $\beta(u,l,t)$ and the positive constant $k$:
\begin{enumerate}
  \item $\gamma > 0$ for  $t\in \ca T$ and  $l\in [l_-,l_+]$.
   \label{con:postivity}
  \item $\exists$ some $\tilde{u}(l,t) \in \ca M$ such that $\beta(\tilde{u},l,t)=0$   $\forall\,l,t$.
  \label{con:zero}
\end{enumerate}
With these, we now proceed by stating and proving the following theorem:
\begin{thm}
Let $F:  [l_-,l_+] \to \ca F$ be a general cost function which is $\ca C^2$ in $[l_-,l_+]$. For the class of control problems satisfying the SDE \eref{eqn:GenSDE}, where $k,\gamma$ and $\beta$ satisfy the constraints (\ref{con:postivity}-\ref{con:zero}), the protocol $\ti u$ is globally optimal (\ie, minimizes the cost function, $F$, at some  final time) if and only if it is also locally optimal for minimizing $F$.
\end{thm}
\begin{proof}
From the SDE \eref{eqn:GenSDE}, it can be shown using \ito calculus that the function $F(l)$ obeys
\begin{eqnarray}
  E[dF_u]  &= k\, E\left [-\gamma l\fd{F}{l} + \gamma l \beta^2  \Big(\fd{F}{l} + \half{l\,k} \sd{F}{l} \Big) \right] dt, \label{eqn:genF}
  \end{eqnarray}
 Then, for a protocol $u^*$ to be locally optimal, it must satisfy
  \begin{equation}
    E[dF_{u^*}] = \inf_u E[dF_u].
  \end{equation}
In particular, for the protocol $\ti{u}$ to be locally optimal, we must have
\begin{equation}
\fd{F}{l} + \half{l\,k} \sd{F}{l} \geq 0 \label{eqn:GenFLopt},
\end{equation}
 which follows from applying the condition of positivity (condition (\ref{con:postivity})) in \Eref{eqn:genF}. We now let $\ti{u}$ be a candidate optimal strategy and show the condition for $\ti{u}$ to be the globally optimal strategy. Firstly, we note that the protocol $\tilde{u}$ is deterministic, since $\beta(\ti{u},l,t)$ vanishes in \Eref{eqn:GenSDE}. Then, the variable $l_{\ti{u}}$ satisfies
  \begin{equation}
    dl_{\ti{u}} = -k \gamma(t) l_{\ti{u}} dt.
  \end{equation}
  Solving the above equation using the initial conditions $l(t_0)=l_0$, we get
  \begin{equation}
    l_{\ti{u}}(t) = l_0 e^{-k \int^t_{t_0} \gamma(t')dt'} \label{eqn:TiuSol}.
  \end{equation}
   Next, using the procedure outlined in the preceding section, we can write down the Bellman equation for global optimality in this problem;
 \begin{equation}
   \fpd{V}{t_0} - k\gamma_0 l_0 \fpd{V}{l_0} + \inf_u \bigg( k \gamma_0 l_0 \beta^2_0 \Big[\fpd{V}{l_0} + \half{ l_0\,k} \spd{V}{l_0} \Big]\bigg) = 0,
 \end{equation}
 where we have abbreviated $\beta_0 = \beta(u,l_0,t_0)$ and $\gamma_0=\gamma(t_0)$, and have explicitly denoted the partial derivatives as with respect to the initial coordinates $l_0$ and $t_0$. Since $\ti{u}$ is deterministic as previously noted, $V(l_0,t_0) = E[F(l)|l(t_0) = l_0]=F(l(t,l_0,t_0))$. We then have the relations
 \begin{eqnarray}
   \fpd{V}{t_0} &= \fd{F}{l} \fpd{l}{t_0}, \label{eqn:rel1}\\
   \fpd{V}{l_0} &= \fd{F}{l} \fpd{l}{l_0}, \label{eqn:rel2}
 \end{eqnarray}
 and so,
 \begin{eqnarray}
      \fpd{V}{t_0} - k \gamma_0 l_0 \fpd{V}{l_0} &= \fd{F}{l} \Bigg( \fpd{l}{t_0} - k \gamma_0 l_0\fpd{l}{l_0}\Bigg),
 \end{eqnarray}
 which is evidently equal to $0$, given \Eref{eqn:TiuSol}. Then, the condition for $\ti{u}$ to be the globally optimal protocol is
 \begin{eqnarray}
   \fpd{V}{l_0} + \half{ l_0\,k} \spd{V}{l_0} &\geq 0, \label{eqn:GenFGopt}
 \end{eqnarray}
 since if \Eref{eqn:GenFGopt} were not satisfied, the infimum would not give $\beta_0 = 0$. Using the relations (\ref{eqn:rel1}) and (\ref{eqn:rel2}), \Eref{eqn:GenFGopt} becomes
 \begin{eqnarray}
   e^{-k \int^t_{t_0} \gamma(t')dt'} \Bigg( \fd{F}{l} + \half{ l\,k} \sd{F}{l} \Bigg) &\geq 0,
    \end{eqnarray}
 which is satisfied if and only if
 \begin{eqnarray}
   \fd{F}{l} + \half{ l\,k} \sd{F}{l}  &\geq 0,
 \end{eqnarray}
 which is evidently the same condition as \Eref{eqn:GenFLopt}. Thus, the protocol $\ti u$ is globally optimal if and only if it is also locally optimal.
\end{proof}
For our original problem, we have $l=L$, $k=4$, $\gamma=1$, and $\beta = u\sqrt{1-2L} $, which can be easily verified to satisfy the constraints (\ref{con:postivity}-\ref{con:zero}). Also, the protocol $\ti{u}$ is $u=0$, which is Jacobs' protocol. Thus, Jacobs' protocol is globally optimal at minimizing the \ren entropy of order $\alpha>0.823$ $\forall\,L,t$.

\section{ Local optimality everywhere does {\em not} imply global optimality in general} \label{sec:WRNotOptimal}
We now focus on the region where $\alpha<0.5$. {As can be seen from Figure \ref{fig:JProofOpt}, the protocol $u=0$ is not locally optimal $\forall\,L\in [0,\half 1]$ and $\alpha<0.5$. As explained previously in the paragraph following \Eref{eqn:LocalOptJorWR}, this implies that the WR protocol $u^2=1$, is locally optimal for this range of parameters.} However, in this section we numerically showed that it does not satisfy \Eref{eqn:BellmanEqn}.   For $u^2=1$, the protocol is no longer deterministic, but we can still obtain the cost for this protocol semi-analytically using linear trajectory theory \cite{Wis96a,JacobsKnight}.  For this null-control protocol \cite{WR}, we have  $x(t)=0$  $\forall\,t$, so that the eigenvalues of $\rho$ are
$\lambda_\pm=(1\pm z)/2$ and $L=2\lambda_-\lambda_+$.
 Then, the expected \ren entropy of order $\alpha$, given $L(t_0) = L_0$ is
\begin{equation}
E(\vS_\alpha(t)|L_0,t_0) = \intall \, \vS_\alpha(L_0,w) \wp(w,L_0,t-t_0) \,dw, \label{eqn:WRCostFun}
\end{equation}
where
\begin{eqnarray}
  \vS_\alpha(L_0,w) &= \inv{1-\alpha} \ln \left[ \pare{\half{1+z(w,L_0)}}^\alpha+\pare{\half{1-z(w,L_0)}}^\alpha\right],\\
  z(w,L_0) &= \frac{\sqrt{1-2L_0}\cosh(2w) + \sinh(2w)}{\cosh(2w)+\sqrt{1-2L_0}\sinh(2w)},
\end{eqnarray}
and the distribution $\wp$ is
\begin{equation}
  \wp(w,L_0,\tau) = \Big(\cosh(2w) + \sqrt{1-2L_0} \sinh(2w)\Big) \frac{e^{-w^2/(2\tau)}}{\sqrt{2\pi \tau}}e^{-2\tau}.
\end{equation}
Using the above equations with $V(L,t) = E(\vS_\alpha(T)|L,t)$, we numerically showed that for $\alpha<0.5$,
\begin{equation}
   \ti{\mathcal{D}} [V] < 0\qquad \forall\,L\;{\rm and }\;t,
\end{equation}
implying that the infimum in \Eref{eqn:BellmanEqn} {gives $u^2=1$ as expected}. However, using $V(L,t)$ given by \Eref{eqn:WRCostFun}, we numerically showed that
\begin{equation} \label{tbref}
  \fpd{V(L,t)}{t} - 4L \fpd{V(L,t)}{L} +  \ti{\mathcal{D}}  [V(L,t)] \neq  0,
\end{equation}
 for any $0< \alpha <0.5$, $L\in [0,\half 1]$ and $\tau \equiv T-t \in(0,3]$. Thus, the protocol $u=1$, despite being locally optimal, does not satisfy the Bellman equation for minimization of the \ren entropies of order $\alpha<0.5$.

\subsection{Understanding the process of verification}
In the previous section, we have shown numerically that the protocol $u(t)=1$ $\forall t\in\ca T$ does not satisfy the Bellman equation $\forall\, t\in\ca T$ for {\em any} value of $\alpha$, even those values for which it is the locally optimal protocol for all $L \in [0,\half1]$.

The Verification Theorem \ref{thm:verify} states that if a solution to the Bellman equation \eref{eqn:BellmanEqn} is found, and some control $u^*$ exists such that it achieves the minimum in \Eref{eqn:minimum}, then the control $u^*$ is the optimal protocol. This theorem does {\em not} imply the inverse \ie~a protocol $u$ and corresponding function $W$ which attains the minimum in \Eref{eqn:minimum} but does not satisfy the Bellman equation cannot be proven not to be the optimal protocol.

To prove the converse, we require the assumptions:
 \begin{enumerate}
   \item The set of admissible control functions defined in \Eref{eqn:SetOfControls} is compact \label{con:Compact}
  \item $\partial^2 W(L,t) / \partial L^2$ is continuous. \label{con:2ndDerCont}
\end{enumerate}
Then, by Lemma 6.1 and Theorem 6.1 of Reference \cite[Chapter VI]{FlemingRishelBook}, a protocol $u$ and corresponding cost function $J$ which does not satisfy the Bellman equation is not the optimal protocol.

 The set of admissible control functions is defined in \Eref{eqn:SetOfControls} as the set of all functions which map the time interval $[t_0,T]$ to $[0,1]$. Without additional constraints we are not confident that the assumption of compactness \eref{con:Compact} applies. Hence we cannot say with certainty that the Wiseman-Ralph protocol  is not an optimal protocol, even though it does not satisfy the Bellman equation \eref{eqn:BellmanEqn}. The question deserves further study.

\section{Conclusion} \label{sec:conclude}
In this paper, we have generalized Jacobs' problem of controlling the measurement-induced purification of a qubit, using the $\alpha$-family of \ren entropies at the final time as the cost function. We have generalized Jacobs' original simple proof of optimality, which enables it to be applied to a finite range, $\alpha\in[0.823,1.103]$, which includes the von Neumann entropy ($\alpha = 1$). Surprisingly,  applying rigorous verification theorems shows that Jacobs' protocol is globally optimal even when it {\em cannot} be shown to be so by his method of proof.  In fact, we showed rigorously, that Jacobs' protocol is globally optimal in the region which it is also locally optimal. This locally optimal region we numerically showed to be $\alpha \in [0.823,50] \cup \{\infty\}$, which led us to conjecture that the full region for which the protocol is locally optimal is $\alpha > 0.823$.

 Our method allows one to use intuitive methods to prove the global optimality of a restricted class of control protocols. Of course such methods do not, in general, prove global optimality. This is a restatement of the fact that the condition of a protocol being locally optimal everywhere is not equivalent to global optimality of the protocol. Indeed, we showed that the Wiseman-Ralph protocol, which is locally optimal for all $L\in [0,\half 1]$ and $\alpha<0.5$, cannot be verified to be the globally  optimal protocol. In future work one could investigate whether it is possible to find protocols that can be verified as globally optimal in this regime, and also in the intermediate regime, $\alpha \in [0.5,0.823]$. One could also compare the expected costs (\ie the purification rate) of different protocols in the whole parameter space of \ren entropies.


\ack
 The authors thank Matthew James and Ian Petersen for helpful discussions regarding the verification theorem. This work was supported by the Australian Research Council Centre of Excellence CE110001027, and the ARO MURI grant W911NF-11-1-0268. JC is supported in part by National Science Foundation Grant Nos. PHY-1212445, PHY-1005540 and PHY-1314763, by Office of Naval Research Grant No. N00014-11-1-0082. CT acknowledges support from the core grant of CQT, a Research Center of Excellence of the National Research Fund and the Ministry of Education of Singapore and hospitality from the Centre of Quantum Dynamics and Griffith University, where this study was performed.

\end{document}